\documentclass[12pt]{article}
\usepackage{amsmath}
\usepackage{graphicx}
\usepackage{url} 
\usepackage[english]{babel}
\usepackage[utf8]{inputenc}
\usepackage{amsmath}
\usepackage{amsfonts}
\usepackage{amsthm}
\usepackage{algorithmic}
\usepackage{mathtools}
\usepackage{amssymb}
\usepackage[colorinlistoftodos]{todonotes}
\usepackage{caption}
\usepackage{subcaption}
\usepackage{float}
\usepackage{mathrsfs}
\usepackage{bbm}
\usepackage{algorithm}
\usepackage{algorithmic}
\usepackage{enumerate}
\usepackage{setspace}   
\usepackage[round]{natbib}
\usepackage{chapterbib}
\usepackage{authblk}

\setcitestyle{aysep={}} 
\newtheorem{thm}{Theorem}
\newtheorem{lem}{Lemma}
\newtheorem{cor}{Corollary}

\newcommand{\RN}[1]{%
	\textup{\uppercase\expandafter{\romannumeral#1}}%
}
\newcommand{\blind}{0}
\DeclareMathOperator*{\argmin}{argmin}

\begin{document}

\def\spacingset#1{\renewcommand{\baselinestretch}%
	{#1}\small\normalsize} \spacingset{1}


\if0\blind
{
	\title{\bf Optimal Sampling for Generalized Linear Models under Measurement Constraints}
\author[1]{Tao Zhang\thanks{email: tz284@cornell.edu}}
\author[1]{Yang Ning \thanks{email: yn265@cornell.edu}}
\author[1,2]{David Ruppert\thanks{email: dr24@cornell.edu}}

\affil[1]{Department of Statistics and Data Science, Cornell University}
\affil[2]{School of Operations Research and Information Engineering, Cornell University}
	\maketitle
} \fi

\if1\blind
{
	\bigskip
	\bigskip
	\bigskip
	\begin{center}
		{\LARGE\bf Title}
	\end{center}
	\medskip
} \fi

\bigskip

\baselineskip=2\baselineskip

\begin{abstract}
\baselineskip=2\baselineskip
Under ``measurement constraints,"  responses are expensive to measure
and initially unavailable on most of records in the dataset, but the covariates are available for the entire dataset.  Our goal is to sample a relatively small portion of the dataset where the expensive responses will be measured and the resultant sampling estimator is statistically efficient.  
Measurement constraints require the sampling probabilities can only depend on a very small set of the responses.	  A sampling procedure that
uses responses at most only on a small pilot sample will be called ``response-free.''
We propose a  response-free sampling procedure \mbox{(OSUMC)} for generalized linear models (GLMs).  Using the A-optimality criterion, i.e.,
the trace of the asymptotic variance, the resultant estimator is statistically efficient within a class of sampling estimators.
 We establish the unconditional asymptotic distribution of a general class of response-free sampling estimators.  This result is novel compared with the existing conditional results obtained by conditioning on both covariates and responses. 
Under our unconditional framework, the subsamples are no longer independent and  new martingale techniques are developed for our asymptotic theory. We further derive the A-optimal response-free sampling distribution. Since this distribution depends on population level quantities, we propose the Optimal Sampling Under Measurement Constraints (OSUMC) algorithm to approximate the theoretical optimal sampling. Finally, we conduct an intensive empirical study to demonstrate the advantages of OSUMC algorithm over existing methods in both statistical and computational perspectives.  We find that OSUMC's performance is comparable to that of sampling algorithms that use complete responses. This shows that, provided an efficient algorithm such as OSUMC is used, there is little or no loss in accuracy due to the unavailability of responses because of measurement constraints.
	\end{abstract}

\noindent%
{\it Keywords:}  Generalized linear models, Measurement constraints, Unconditional asymptotic distribution, Martingale central limit theorem, A-optimality 
\vfill

\setcounter{section}{0} 
\setcounter{equation}{0} 

\section{Introduction}
\subsection{Motivation and Contribution}
Measurement constrained datasets \citep{WangY2017} where only a small portion of the data points have known a response, $Y$, but the covariates, $X$,  are available for all data points, are common in practice. Datasets of this type can happen when the response  is more expensive or time-consuming to collect than covariates.  We will present two motivating examples. For more real-world examples, we refer readers to semi-supervised learning literature~(e.g., \citealt{Zhu2005,Chapelle2006,Chakrabortty2018}).

\bigskip
\noindent
\textbf{Example 1.}  \textbf{(Critical Temperature of Superconductors)} \quad Critical temperature, which is sensitive to chemical composition, is an important property of superconducting materials. Since no scientific model for critical temperature prediction is available~\citep{Hamidieh2018}, a data-driven prediction model is desirable to guide researchers synthesizing superconducting materials with higher critical temperature. Due to the cost in both money and time for material synthesis, only a small portion of the thousands of potential chemical compounds can be manually tested. So selecting representative compositions to build a statistical model with maximum efficiency is important. \\
\noindent
\textbf{Example 2.}  \textbf{(Galaxy Classification)} \quad Galaxy classification is an important task in astronomy~\citep{Banerji2010}.  Visual classification is time-consuming and expensive and is becoming infeasible because the size of astronomical datasets is growing rapidly as more advanced telescopes enter operation. The size of modern galaxy datasets is often of millions or even billions~\citep{Reiman2019}.  It is important to select a representative subsample of galaxies which can be classified accurately by humans, so that an effective classification model can be built.

\bigskip
In addition to many responses being missing, another characteristic of the datasets in these examples is the extreme size which brings huge challenges to statistical computing, data storage and communication. Sampling is a popular approach to super-large datasets where a small portion of the dataset is sampled and used as a surrogate of the entire dataset.  There is a large literature on this problem, e.g., \cite{Wang2017}, but many of the proposed sampling methods assume that the response is known on the entire dataset  so are not applicable under measurement constraints.  Moreover, they might not increase statistical efficiency even if they were applicable. 

The problem addressed in the paper is fitting a generalized linear model (GLM) efficiently to massive measurement-constrained datasets using only a relatively-small subsample obtained by sampling. 
Measurement constrained datasets require sampling methods to be \emph{response-free} meaning that the sampling probabilities can only depend on responses of  a small pilot sample.
Though a huge literature has been devoted to sampling methods for datasets where all records have both responses and covariates, efficient sampling algorithms under measurement constraints, which selects subsamples in a (nearly) response-independent way, are less well-studied. 
We bridge this gap by proposing a statistically efficient sampling method under measurement constraint (OSUMC) for GLMs which does not require knowledge of responses (except in the pilot sample). Specifically, OSUMC selects subsamples with replacement according to a response-free A-optimal sampling distribution and computes the  estimator by solving a weighted score equation based on the selected subsamples. We defer the details to Section \ref{gss} and Section \ref{thr}. 

This paper contributes to both theoretical and computational perspectives of the growing literature on statistical sampling. On the theoretical side, we prove asymptotic normality of a general class of response-free sampling estimators without conditioning on the data.
Our asymptotic results are significantly different from conditional asymptotics in the traditional sampling literature which condition on both covariates and responses.  
 In particular, the conditional independence of subsamples in the conditional asymptotics is no longer true under our framework. To deal with the correlation in our sampling estimator, we develop novel martingale techniques for our asymptotic theory.  Another significant difference is that our asymptotic results compare our sampling estimator to the true parameter rather than to the maximum likelihood estimator as done by \cite{Wang2017}.  Since the true parameter is the object of interest, our results are more informative. 
 Based on the asymptotic theory, we derive the optimal sampling probabilities by minimizing the A-criterion \citep{Khuri2006} i.e., the trace of asymptotic variance matrix, which is equivalent to minimizing the sum of the asymptotic mean squared errors. On the computational side, we propose the OSUMC algorithm to approximate the theoretical A-optimal sampling method.  Our sampling algorithm achieves  an \emph{optimal design} by assigning higher probabilities to data points that achieve estimators with higher efficiency.  We show in our numerical study that the performance of OSUMC is comparable to that of sampling algorithms which use complete responses to calculate sampling probabilities.  This shows that, provided that an efficient algorithm such as OSUMC is used, there is little or no loss in accuracy due to not having the responses available for sampling because of measurement constraints.

\subsection{Related Work}
A large literature provide numerous variants of subsampling algorithms for linear regression \citep{Drineas2006,Drineas2011,Ma2015,wang2019information}. One traditional approach is the leverage sampling 
which defines the sampling probabilities based on the empirical leverage scores of the design matrix \citep{Huber2011}.
In more recent literature, leverage sampling serves as an important basis for more refined sampling methods of linear regression \citep{Drineas2006,Drineas2011,Ma2015}. 
These papers fall in a paradigm termed \emph{algorithmic leveraging}, which is fundamental to randomized numerical linear algebra (RandNLA) which aims at developing fast randomized algorithms for large-scale matrix-based problems. We refer readers to \cite{Ma2015} and \cite{drineas2016randnla} for more references. However, most of the algorithmic leveraging literature is concerned with the numerical performance of algorithms and only a few papers, for example, \cite{Ma2015} and \cite{raskutti2016statistical}, provide statistical guarantees. 
\cite{Ma2015} derives bias and variance for their sampling methods while \cite{raskutti2016statistical} provide statistical error bounds for their sketching method.
 In contrast, we will focus on the asymptotic efficiency which 
has not been considered in the algorithmic leveraging or RandNLA literature.  Because leverage sampling for the least-square estimator is response-free it can adapt to measurement constraints, but this is not true for other GLMs. 

For other GLMs, \cite{Wang2017} developed an A-optimal sampling procedure for logistic regression. Though a similar optimality criterion is adopted, our paper is very different from \cite{Wang2017}. Most importantly, the sampling procedure in \cite{Wang2017} requires the complete set of responses and therefore cannot be used under measurement constraints, while OSUMC is tailored for that setting. From a theoretical perspective, the unconditional asymptotics derived in this paper are much more challenging than the conditional asymptotic setting considered in \cite{Wang2017}, which conditions on both covariates and responses. The essential conditional independence assumption in the proof of \cite{Wang2017} is no longer valid and the correlated structure of the sampling estimator must be treated in our asymptotic theory. New martingale techniques are developed which are significantly different from the techniques used in \cite{Wang2017}. Besides, unconditional framework allows us to prove the asymptotic results between our sampling estimator and the true parameter, instead of the conditional MLE as in \cite{Wang2017}. A recent paper by \cite{ai2018optimal} generalized the results in \cite{Wang2017} to other GLMs under a similar framework, so again not applicable to the problem of measurement constraints. In another paper by \cite{Ting2018}, the authors studied optimality of sampling for asymptotic linear estimators. Their conclusions will reduce to exactly the same results in \cite{Wang2017} for logistic regression, and hence cannot deal with measurement constraints neither.

One new research area also dealing with measurement constraints is \emph{semi-supervised learning} (SSL). SSL attempts to use the unlabeled $X$ ($X$ without $Y$) to improve statistical performance. Though a huge SSL literature is devoted to algorithmic aspects~\citep{Zhu2005,Chapelle2006}, only a few recent papers studied statistical estimation under the semi-supervised setting~\citep{zhang2016semi,cai2018semi,Chakrabortty2018}. For example, \cite{Chakrabortty2018} estimate linear regression coefficients by regressing imputations of unobserved responses on the corresponding covariates.
 However, their method is computationally prohibitive for large datasets due to the nonparametric imputation approach while our method is computationally affordable even for massive datasets as demonstrated in out empirical study.
  
Another closely related area is \emph{optimal experiment design}  which determines the settings of covariates that yield estimators with optimal properties \citep{Khuri2006}.  Whereas the design is determined  freely in classical experiment design~\citep{Pukelsheim2006}, the design points of this paper must be selected from the original dataset. Also, the solutions of the traditional optimal design problem are often combinatorial, which are computationally infeasible for even moderate size datasets. \cite{WangY2017} propose sampling algorithms based on the convex relaxation of the traditional combinatorial problem.  However, \cite{WangY2017} mainly focus on algorithms and only a 
mean square error bound is provided for the statistical guarantee. Their results are also proved under traditional conditional framework which is different from the unconditional analysis of this paper.   Additionally, independence assumptions for with-replacement sampling in \cite{WangY2017} cannot be justified in many real situations, for instance, the two motivating examples. Such assumptions are avoided in our theory. On the practical side, our method offers simple closed forms of the optimal sampling weights while \cite{WangY2017}'s procedure involves solving a semi-definite programming problem which can be computationally intensive for large datasets. 
\subsection{Organization of the Paper}
The remainder of the paper is organized as follows. In Section \ref{bg}, we give a brief overview of GLMs. In Section \ref{gss}, we formulate a general response-free sampling scheme and define the class of sampling estimators of interest. In Section \ref{thr}, we derive the asymptotic normality of our sampling estimator and
define the A-optimal sampling distribution. OSUMC algorithm is proposed to approximate the theoretical optimal sampling procedure. Section \ref{sr} compares OSUMC algorithm with several existing sampling methods on both synthetic and real data, where linear and logistic regression models are mainly discussed. We summarize the paper in Section \ref{con}. Proofs are in the supplementary materials. 

\section{Background and Setup }\label{bg}
We begin with the background on GLMs.   Assume $n$ independent and identically distributed data couples, $(X_1,Y_1),\dots,(X_n,Y_n)\sim (X,Y)$, where $X\in \mathbb{R}^p$ is a  covariate vector, $Y\in \mathbb{R}$ is the response, and $Y$ given $X$ satisfies a GLM with the canonical link:
$$P(Y|X,\beta_0,\sigma)\propto \exp\left\{\frac{Y\cdot X^T\beta_0- b(X^T\beta_0)}{c(\sigma)} \right\},$$
Here $b(\cdot)$ is a known function, $\sigma$ is a known dispersion parameter, and $\beta_0\in \mathbb{R}^p$ is the unknown parameter of interest and assumed to be in a compact set $\mathcal{B}\subseteq \mathbb{R}^p$.
Without loss of generality, we take $c(\sigma)=1$.   
The standard estimator of $\beta_0$ is the MLE
$$\hat{\beta}_{\rm mle}\in \argmin_{\beta\in \mathbb{R}^p}\left[- \frac{1}{n} \sum_{i=1}^n\left\{Y_i\cdot X_i^T\beta-b(X_i^T\beta)\right\} \right].$$
Equivalently, we could solve the score equation to obtain the MLE
$$\Psi_n(\beta):=\frac{1}{n}\sum_{i=1}^n\psi_{\beta}(X_i,Y_i)=\frac{1}{n} \sum_{i=1}^n\left\{b'(X_i^T\beta)-Y_i\right\}\cdot X_i=0.$$
Iterative methods such as Newton's method and its variants are usually adopted to solve such problems numerically \citep{mccullagh1989generalized,aitkin2005statistical}. If the sample size $n$ is very large,  the computational cost for just one iteration will be huge. Therefore, a sampling approach can be used to reduce computational cost. 

In addition, we assume a measurement constraint setting where only a small portion of responses are available initially. As mentioned in the introduction, this setting is common in practice. The main purpose of this paper is to develop a unified and statistically efficient sampling procedure for GLMs under measurement constraints.

\section{General Sampling Scheme}\label{gss}
We first present in Algorithm \ref{A1} a general response-free sampling scheme for GLMs. The class of response-free sampling estimators is defined accordingly.
\begin{algorithm}[h]

\smallskip
	\begin{enumerate}
		\item Sample with replacement from the original $n$ data points $r$ times with probabilities $\pi=\{\pi_i\}_{i=1}^n$, where we require that $\pi_i$ only depends on $(X_1,\dots,X_n)$ and a pilot estimate of $\beta$, but not $(Y_1,\dots,Y_n)$. Collect the subsample ${(X_i^*,Y_i^*)}_{i=1}^r$, where we let  $(X_i^*,Y_i^*)$ denote the data sampled out in the $i$-th step.
		\item Define the reweighted score function as
		$$ \Psi_n^*(\beta):=\frac{1}{r} \sum_{i=1}^r\frac{b'(X_i^{*T}\beta)-Y_i^*}{n\pi_i^*}\cdot X_i^*$$
		where $\pi_i^*$ corresponds to the sampling probability of $(X_i^*,Y_i^*)$.
		\item Solve the reweighted score equation $\Psi^*_n(\beta)=0$ to get the estimator $\hat{\beta}_n$. 
	\end{enumerate}
	\caption{Response-free sampling procedure for GLMs}\label{A1}
\end{algorithm}

 We emphasize that Algorithm \ref{A1} is tailored for the measurement constraints setting, as the sampling weight $\pi_i$ only depends on $(X_1,\dots,X_n)$ not $(Y_1,\dots,Y_n)$ which may not be completely observable. In practice, once the subsample of size $r$ are selected, one need measure only those $r$ responses. This may bring huge economic savings as response measurements are usually expensive in measurement constrained situations.
 
 An additional benefit of sampling is cost savings.
  In the last step of Algorithm~\ref{A1}, Newton's method and its equivalent variants are usually adopted \citep{mccullagh1989generalized, aitkin2005statistical}. Given the sampling probabilities $\pi$, Algorithm \ref{A1} dramatically reduces the computational and storage costs by making them scale in $r$ instead of $n$ which could be much larger than $r$. More concretely, if $n=10^6$ and $p= 20$, the computational time of one iteration of Newton's method for the full sample MLE is $O(np^2)=O(4\times10^8)$.  In addition, if each data point occupies 1MB of storage space, then the original dataset would occupy around 1TB space. In contrast, for Algorithm \ref{A1} with $r=1,000$,  the computational time for each iteration is $O(4\times 10^5)$ and the storage space is less than 1 GB, which substantially lower the computational and storage cost. 
 
 The performance of Algorithm \ref{A1} depends crucially on the choice of the sampling weights $\pi_i$  and the subsample size $r$.   Under measurement constraints, subsample size $r$ is determined by the cost of measuring the response and, perhaps, by the availability of the computational or storage resources. The more important question is how to determine the sampling distribution in a data-driven approach such that the resultant sampling estimator achieves optimal efficiency. We will answer this question in the next section by defining the A-optimal sampling distribution based on the asymptotic results therein.

\section{Optimal Sampling Procedure and Asymptotic Theory}\label{thr}
In this section, we assume the classical asymptotic setting in which $n\to \infty$ and $p$ is fixed. 
We first show the asymptotic normality of the sampling estimator defined in Algorithm~\ref{A1}, and then find the A-optimal sampling distribution by minimizing the asymptotic mean squared error.  
\subsection{Notation}
The $j$th entry of the covariate vector $X_i$ is denoted by $x_{ij}$. For $X \in \mathbb{R}^p$, $||X||$ is the Euclidean norm of  $X$. We also define tuple notations: $X_1^n:=(X_1,X_2,\dots,X_n)$ and $Y_1^n:=(Y_1,Y_2,\dots,Y_n)$. $V(X)$ and $E(X)$ denote the variance and expectation of $X$, respectively.
\subsection{Consistency of $\hat{\beta}_n$}
We first show the statistical consistency of $\hat{\beta}_n$.
\begin{thm}[Consistency of $\hat{\beta}_n$]\label{Con}
 Assume the following conditions	
	\begin{enumerate}[(i)]
		\item Either $(ia)$ $b^{''}(\cdot)$ is bounded or $(ib)$ X is bounded.  \label{Con1}
		\item $EXX^T$ is finite and $\zeta(\beta):=E\left[\{b'(X^T\beta)-Y\}X\right]$ is finite for any $\beta\in \mathcal{B}$. \label{Con2}
		\item $\sum_{i=1}^n E\left[ \frac{\{b'(X_i^T\beta)-Y_i\}^2}{\pi_i} x_{ij}^2\right]=o(n^2r)$ for $1\le j\le p$ and $\beta \in \mathcal{B}$. \label{Con3}
		\item $\inf_{\beta:||\beta-\beta_0||\ge \epsilon}||\zeta(\beta)||>0$ for any $\epsilon>0$. \label{Con4}
	\end{enumerate}
 Then $\hat{\beta}_n\mathop{\to}\limits^{p} \beta_0$.
\end{thm}
Condition (\ref{Con1}$a$) is satisfied for most of GLMs except for Poisson regression. For Poisson regression, our theory is still applicable if Condition (\ref{Con1}b) is satisfied.  Condition (\ref{Con3}) - (\ref{Con4}) are needed to apply a consistency theorem for M-estimators \citep{Van2000}. In particular, condition (\ref{Con3}) ensures the uniform convergence of $\Psi_n^*$ while condition (\ref{Con4}) is a common \emph{well-separated} condition for consistency proofs which is satisfied if $\zeta(\cdot)$ has unique minimizer and if, as we have assumed, the parameter space $\mathcal{B}$ is compact.

\subsection{Asymptotic Normality of $\hat{\beta}_n$}
To establish asymptotic normality of $\hat{\beta}_n$, we start with an important asymptotic representation.
\begin{lem}[Asymptotic linearity]\label{AL}
	  Assume the following conditions,
	\begin{enumerate}[(i)]
		\item $\Phi=E\left\{b^{''}(X^T\beta_0)XX^T\right\}$ is finite and non-singular.
		\item $\sum_{i=1}^{n}E\left\{\frac{b^{''}(X_i^T\beta_0)^2}{\pi_i}(x_{ik}x_{ij})^2\right\}=o(n^2r)$, for $1\le k,j \le p$.
		\item $b(x)$ is three-times continuously differentiable for every $x$ within its domain.
		\item Every second-order partial derivative of $\psi_{\beta}(x)$ w.r.t $\beta$ is dominated by an integrable function $\ddot{\psi}(x)$ independent of $\beta$ in a neighborhood of $\beta_0$.\label{AL3}
	\end{enumerate}
	If $\Psi_n^*(\hat{\beta}_n)=0$ for all large $n$ and if $\hat{\beta}_n$ is consistent for $\beta_0$, then 
	$$\Psi_n^*(\beta_0)=-\Phi(\hat{\beta}_n-\beta_0)+o_p(||\hat{\beta}_n-\beta_0||).$$	
\end{lem}

The proof of Lemma \ref{AL} is based on the asymptotic linearity of M-estimators, e.g., Theorem~5.41 in \cite{Van2000}. It turns out to be important for the establishment of asymptotic normality.

We will apply a multivariate martingale central limit theorem (Lemma \ref{MCLT} in the supplementary materials) to the above asymptotic linear representation and show the asymptotic normality of $\hat{\beta}_n$. 

We first define a filtration $\{\mathcal{F}_{n,i}\}_{i=1}^{r(n)}$ adaptive to our sampling procedure: $\mathcal{F}_{n,0}=\sigma(X_1^n,Y_1^n)$; $\mathcal{F}_{n,1}=\sigma(X_1^n,Y_1^n)\lor\sigma(*_1)$; $\cdots$; $\mathcal{F}_{n,i}=\sigma(X_1^n,Y_1^n)\lor\sigma(*_1)\lor\cdots \lor \sigma(*_i)$; $\cdots$, where $\sigma(*_i)$ is the $\sigma$-algebra generated by $i$th sampling step, which can be interpreted as the smallest $\sigma$-algebra that contains all the information in $i$th step. In the following, we always assume subsample size $r$ is increasing with $n$. Based on the filtration, we define the martingale 
\begin{align*}
M&:=\sum\limits_{i=1}^{r}M_i:= \sum\limits_{i=1}^{r}\left[ \frac{b'(X_i^{*T}\beta_0)-Y_i^*}{rn\pi_i^*}\cdot X_i^*-\frac{1}{rn}\sum\limits_{j=1}^n\{b'(X_j^T\beta_0)-Y_j\}\cdot X_j\right],
\end{align*}
where $\{M_i\}_{i=1}^r$ is a martingale difference sequence adapt to $\{\mathcal{F}_{n,i}\}_{i=1}^{r}$. In addition, we define: $Q:=\frac{1}{n}\sum_{j=1}^n(b'(X_j^T\beta_0)-Y_j)\cdot X_j$; $
T:=\Psi_n^*(\beta_0)=M+Q$; $
\xi_{ni}:=V(T)^{-\frac{1}{2}}M_i$; $
B_n:= V(T)^{-\frac{1}{2}}V(M)V(T)^{-\frac{1}{2}}$, which is the variance of the normalized martingale $V(T)^{-\frac{1}{2}}M$.

\begin{thm}[Asymptotic normality of $\hat{\beta}_n$]\label{AN}
	Under the conditions in Lemma \ref{AL} and we further assume
	\begin{enumerate}[(i)]
		\item $\lim\limits_{n\to \infty} \sum\limits_{i=1}^rE[||\xi_{ni}||^4]=0$, \label{AN1}
		\item $\lim\limits_{n\to \infty}E[||\sum\limits_{i=1}^{r}E[\xi_{ni}\xi_{ni}^T|\mathcal{F}_{n,i-1}]-B_n||^2]=0$, \label{AN2}
	\end{enumerate}
	we have
	$$V(T)^{-\frac{1}{2}}\Phi(\hat{\beta}_n-\beta_0)\mathop{\longrightarrow}\limits^d N(0,I). $$	
\end{thm}
 Condition (\ref{AN1}) and (\ref{AN2}) are martingale version integrability conditions similar to  Lindeberg-Feller conditions. In fact, such conditions are common in martingale central limit theorems \citep{hall2014martingale}.

\subsection{Optimal Sampling Weights under Measurement Constraints}
In this section, we will derive the A-optimal sampling distribution for our general response-free sampling procedure. 

Theorem \ref{AN} shows that for sufficiently large subsample size, the distribution of $\hat{\beta}_n-\beta_0$ can be well approximated by $N\left(0,\mathbf{V} \right)$ with $\mathbf{V}:=\Phi^{-1} V(T)\Phi^{-1}$. If $\beta_0$ is univariate, we can optimize the sampling probability by minimizing the asymptotic variance $\Phi^{-1} V(T)\Phi^{-1}$ which results in highest statistical efficiency. For multi-dimensional $\beta_0$, 
we adopt the A-optimality criterion of experiment design \citep{Kiefer1959} and minimize the trace of the covariance matrix. Minimization of the trace of $\mathbf{V}$, i.e. ${\rm tr}(\Phi^{-1} V(T)\Phi^{-1} )$, is equivalent to minimization of the asymptotic mean squared error. The following theorem specifies A-optimal sampling distribution. 
\begin{thm}\label{thm2}
	If for $1\le j \le n$, the sampling probability is 
	$$\pi_j =\frac{\sqrt{b^{''}(X^T_j\beta_0)}||\Phi^{-1}X_j||}{\sum_{i=1}^n\sqrt{b^{''}(X^T_j\beta_0)}||\Phi^{-1}X_j||} ,$$
	then ${\rm tr}(\mathbf{V})$ will attain its minimum, i.e., $\{\pi_j\}_{j=1}^n$ is the A-optimal sampling distribution. 
\end{thm} 
The optimal weights cannot be calculated directly in practice, since they depend on population level quantities $\Phi^{-1}$ and $\beta_0$. Therefore, to implement response-free sampling, we need pilot estimates of $\Phi$ and $\beta_0$. The details are shown in Algorithm \ref{ts}. 
\begin{algorithm}[H]
	\begin{enumerate}
	\medskip
		\item Uniformly sample $r_0$ ($\ll$ $ r$) data points with indices $i_1,\dots,i_{r_0}$ and collect those points: $\{(X_{i_j},Y_{i_j})\}_{j=1}^{r_0}$ from data pool. Calculate $\tilde{\beta}_n$, the pilot estimator of $\beta_0$, and  $\tilde{\Phi}_n:=\frac{1}{r_0}\sum\limits_{j=1}^{r_0} b^{''}(X_{i_j}^T\tilde{\beta}_n)X_{i_j}X_{i_j}^T$,  the pilot estimator of $\Phi$, based on the $r_0$ data points.
		\item Calculate the approximate optimal sampling weight for each data point: 
		$$\pi_j \propto \sqrt{b^{''}(X^T_j\tilde{\beta}_n)}||\tilde{\Phi}_n^{-1}X_j||,$$
		for $1\le j\le n$.
		\item Run Algorithm \ref{A1} with $\pi_j$ defined above to obtain the final estimator $\hat{\beta}_n$
	\end{enumerate}
	\caption{Optimal Sampling under Measurement Constraint (OSUMC)}\label{ts}
\end{algorithm}

\textbf{Remarks}
\begin{itemize}
	\item Step 1 in Algorithm \ref{ts} is designed for pilot estimation when very few or even none of the responses are available initially, but expensive responses collection is possible. If a moderate number of responses are accessible in the initial data pool, we could use them to calculate pilot estimators.
	Otherwise, a small random sample can be taken with uniform sampling probabilities.   The size of the pilot sample, $r_0$ could be fairly small compared with total sample size $n$ and even of the same order of magnitude as the dimension $p$ of the problem. In our empirical study, we set $r_0=500$ with the complete dataset of size $n=10^5$ and $p$ ranging from $20$ to $100$, and the performance of Algorithm \ref{ts} is satisfactory. 
	
	\item (Computational complexity) When we use Newton's method, or one of its variants such as Fisher scoring, to compute the root of the score equation, the computation requires $O(\zeta n p^2)$ computational time, where $\zeta$ is the number of iterations needed for the algorithm to converge. In our empirical study, $\zeta$ varies from $10$ to $30$ under different models. For OSUMC algorithm the first step requires $O(\zeta_1 r_0p^2)$ computation time where $\zeta_1$ is the number of iterations. In the second step, $O(np^2+\zeta_2rp^2)$ computation time is required where $\zeta_2$ is the number of the iterations in this step.
	 Hence, OSUMC algorithm has complexity of order $O(np^2+\zeta_1 r_0p^2+\zeta_2rp^2)$. If $n$ is extremely large such that $p$, $r_0$, $r$, $\zeta_1$ and $\zeta_2$ are all much smaller than $n$, the computation complexity of the algorithm is $O(np^2)$. Therefore,  the computational advantage of using OSUMC algorithm compared with full-sample MLE is still huge if the scale of the problem is large, i.e., $np^2$ is large and $\zeta >10$. The intensive numerical study in the following provides strong evidence for such advantage. To further reduce the computational complexity, one may use modified Newton's method for large-scale computation; see, for example, \cite{xu2016sub}.
\end{itemize}

\section{Numerical Examples}\label{sr}
\subsection{Simulation Results}
In this section, we evaluate the performance of the OSUMC algorithm on synthetic datasets. Due to page limitation, we will show the numerical results for logistic and linear regression and defer Poisson regression to Section \ref{Poisson} in the supplementary material. All the results are obtained in the R environment with one Intel Xeon processor and 8 Gbytes RAM over Red Hat OpenStack Platform.

\subsubsection{Logistic Regression}
We generate datasets of size $n=100,000$ from the  logistic regression model,
$$P(Y=1|X,\beta_0)=\frac{\exp(X^T\beta_0)}{1+\exp(X^T\beta_0)},$$
where $\beta_0$ is a $20$ dimensional vector with all entries $1$. We consider four different scenarios to generate $X$ as in \cite{Wang2017}. 
\begin{itemize}
	\item \textbf{mzNormal.} $X$ follows the multivariate normal distribution $N(0,\Sigma)$ with $\Sigma_{ij}=0.5^{I(i\ne j)}$. In this case, we have a balanced dataset, i.e., the number of 1's and the number of 0's in the responses are almost equal. 
	\item  \textbf{nzNormal.} $X$ follows the multivariate normal distribution $N(0.5,\Sigma)$. In this case, we have an imbalanced dataset where about $75\%$ of the responses are 1's.
	\item \textbf{unNormal.} $X$ follows the multivariate normal distribution with mean zero but different variances. To be more specific, $X$ follows the multivariate normal distribution $N(0,\Sigma_1)$ with $\Sigma_1=U_1\Sigma U_1$, where $U_1=diag(1,1/2,\dots,1/20). $
	\item \textbf{mixNormal} $X \sim 0.5N(0.5,\Sigma)\ +\  0.5N(-0.5,\Sigma)$.
\end{itemize}
In each case, we compare our optimal sampling procedure (OSUMC) with the mMSE method in \cite{Wang2017} (OSMAC), uniform sampling (Unif), and the benchmark full data MLE under different subsample sizes ($r$). In our procedure, we set the subsample size of uniform sampling in the first step equal to $r_0=500$. For uniform sampling, we directly subsample $r$ points and calculate the subsample MLE.

We repeat simulations $S=500$ times, and calculate the empirical MSE as  $S^{-1}\sum_{s=1}^S||\hat{\beta}_n^{(s)}-\beta_0||^2$ where $\hat{\beta}_n^{(s)}$ is the estimate from the $s$th repetition. The comparison of the empirical MSE is presented in Figure \ref{logmse}.

\begin{figure}[h]
	
	\centering
	\includegraphics[scale=0.52]{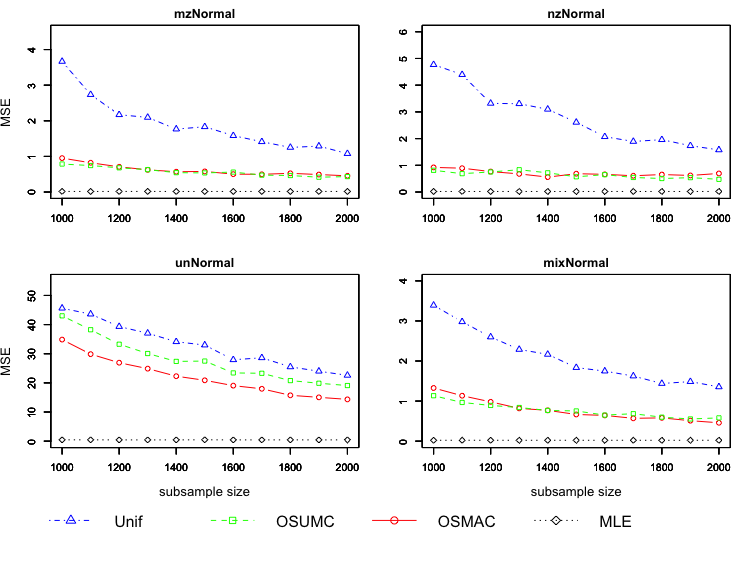}
	\caption{MSE of the proposed optimal sampling procedure (OSUMC), the method in \cite{Wang2017} (OSMAC), the uniform sampling (Unif), and the full sample MLE (MLE) for different subsample size $r$ under four scenarios in logistic regression.}
	\label{logmse}
\end{figure}

From Figure \ref{logmse}, both OSUMC method and the OSMAC method in \cite{Wang2017} uniformly dominate the uniform sampling method in all four scenarios, which agrees with Theorem \ref{thm2}. In most of the simulation settings (except for unNormal), our sampling procedure performs similarly to the OSMAC in \cite{Wang2017}. This is because both methods adopt the A-optimality criterion in respective framework to derive the sampling weights. However, we note that OSMAC requires the response of each data point, which is infeasible under measurement constraints, while our method can be implemented as long as a moderate number of responses are available for the pilot estimators.

We also compare the average computational time for each method under all scenarios and the computational time plot can be found in Section \ref{logtm} in the supplementary materials. Our simulation reveals that the computation time is not very sensitive to the subsample size for all the four methods. In most cases, OSUMC and OSMAC perform similarly and require significantly less computational time compared with the full-sample MLE. 

To see whether the asymptotic normality in our theory holds under the previous four different design generation settings, we plot the chi-square Q-Q plot of the resultant estimator $\hat{\beta}_n$ based on $1000$ simulations with fixed subsample size $r=5000$ for each considered setting. The plots are presented in the supplementary materials, Section \ref{logqqs}. Q-Q plots reveal that the resultant sampling estimator $\hat{\beta}_n$ is approximately normal with sufficiently large sample size $n$ and subsample size $r$ in the four considered design generation settings.

\subsubsection{Linear Regression}
We generate datasets of size $n=100,000$ and dimension $p=100$ from the following linear regression model: $Y=X\beta_0+\epsilon$, 
where $\beta_0=(\underbrace{0.1,\dots,0.1}_{5},\underbrace{10,\dots,10}_{90},\underbrace{0.1,\dots,0.1}_{5})^T$ and $\epsilon\sim N(0,9I_n)$. 

We note that in linear regression model, OSUMC algorithm is equivalent to the following Algorithm \ref{lralgo}. Algorithm \ref{lralgo} is similar to the general least-squares sampling meta-algorithm in \cite{Ma2015}, which is adopted in \cite{Drineas2006,Drineas2011,Drineas2012}. 
\begin{algorithm}[H]
	
	\medskip
	\begin{enumerate}
		\item Uniformly sample $r_0(\ll r)$ data points: $\{(X_{i_j},Y_{i_j})\}_{j=1}^{r_0}$ . Calculate\\ $\tilde{\Phi}_n:=\frac{1}{r_0}\sum\limits_{j=1}^{r_0} X_{i_j}X_{i_j}^T$,  the pilot estimator of $\Phi$.
		\item Calculate the approximate optimal sampling weight for each data point: 
		$$\pi_j \propto ||\tilde{\Phi}_n^{-1}X_j||,$$
		for $1\le j\le n$. 
		\item Repeat sampling $r$ times according to probability in step $2$ and rescale each sampled data point $(X_i^*,Y_i^*)$ by $1/\sqrt{\pi_i^*}$, $1\le i \le r$.
		\item Calculate the ordinary least-squares estimator of the rescaled subsample and output it as the final estimator.
	\end{enumerate}
	\caption{Optimal Sampling for Linear Regression}\label{lralgo}
\end{algorithm}
We consider the following design generation settings from \cite{Ma2015}. Similar settings are also investigated in \cite{WangY2017}.
\begin{enumerate}
	\item \textbf{GA.} The $n\times p$ design matrix $\mathbf{X}$ is generated from multivariate normal $N(1_p, \Sigma_2)$ with $\Sigma_2=U_2\Sigma U_2$, where  $U_2={\rm diag}(5,5/2,\dots,5/30) $.  
	\item  $\mathbf{T_3.}$ Design matrix $\mathbf{X}$ is generated from multivariate t-distribution with 3 degrees of freedom and covariance $\Sigma_2$ as GA.
	\item  $\mathbf{T_1.}$ Design matrix $\mathbf{X}$ is generated from multivariate t-distribution with 1 degrees of freedom and covariance $\Sigma_2$ as GA.
\end{enumerate} 

We compare our method (OSUMC) with uniform sampling (Unif), leverage sampling (Leverage) and shrinkage leveraging (SLEV) in \cite{Ma2015} over different subsample size in the three design generation settings above. For the SLEV method, the shrinkage parameter $\alpha$ is set to be $0.9$ as in \cite{Ma2015}. Again, we repeat the simulation $500$ times and report the empirical MSE and computational time in Figures \ref{lrmse} and \ref{lrtime}, respectively.  

\begin{figure}[h]
	\centering
	\includegraphics[scale=0.42]{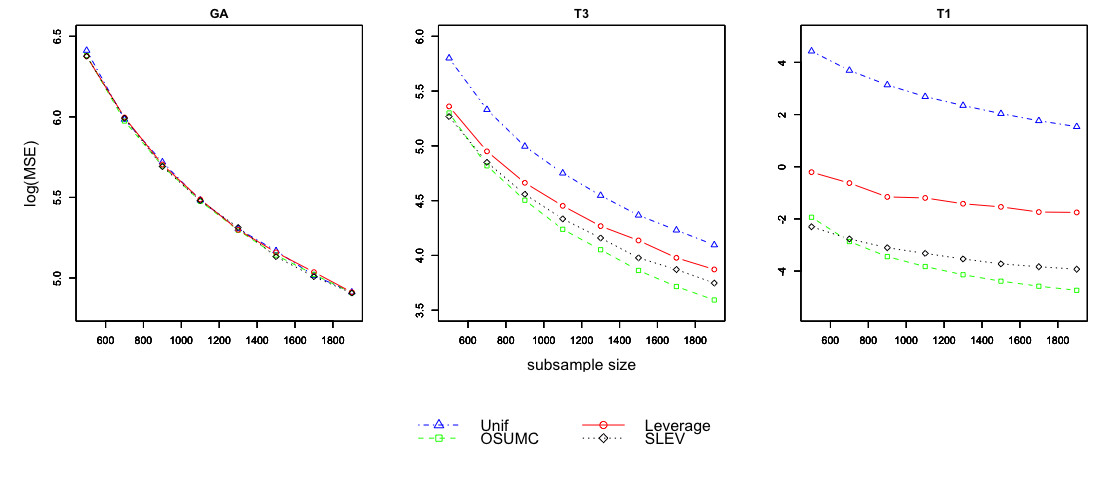}
	\caption{MSE plots for different subsample size $r$ under different design generation settings for linear regression}
	\label{lrmse}
\end{figure}
For all three design generation settings, our method always results in smaller MSE than the other three methods, which again is consistent with our theoretical results. The advantage of our method becomes more obvious when the design generation distribution is more heavy-tailed by noting that $GA$ has moments of arbitrary orders while $T_k$ only has moments up to order $k-1$. It is interesting to see that our method outperforms  other methods even in the $T_1$ design setting where the moment assumptions imposed in our theory are violated. The performance of both leverage sampling and shrinkage leverage sampling improves with heavier-tailed design generation distributions. This has been well understood in the literature on leverage sampling and outlier diagnosis~\citep{Rousseeuw2011,Ma2015}. As expected, all the three sampling methods above yield smaller MSEs than uniform sampling.

The average computational times of the four methods are reported in Figure \ref{lrtime} in the supplementary materials. Again, the results show the insensitivity of the computational time to increasing subsample sizes. Our method requires the second smallest computing time, being inferior only to the uniform sampling. Both leverage related methods take more than double the computational time of our method due to the intensive computation of leveraging score of each data point.

Again, we provide Q-Q plots of the resultant estimator $\hat{\beta}_n$ for each considered design setting and the results can be found in Section \ref{lrqqs} in the supplementary materials.

\subsection{Superconductivty Dataset}
In this section, we analyze the superconductivty dataset \citep{Hamidieh2018}, which is available from the Machine Learning Repository at: \url{https://archive.ics.uci.edu/ml/datasets/Superconductivty+Data#}. The purpose of \cite{Hamidieh2018} is to build a statistical model to predict the superconducting critical temperature of superconducting materials based on their chemical formulas. In the dataset, 21,263 different superconductors' critical temperatures are collected along with $81$ features extracted from the chemical formulas the superconductors. Multiple linear model is considered in \cite{Hamidieh2018} and regression coefficients are calculated based on the full sample, which is treated as ``true" parameters ($\beta_0$) in the following analysis. 

We compare our sampling method (OSUMC) with the three other sampling methods in the simulation studies for linear models. Besides the estimation accuracy which is the main focus before, prediction performance of the sampling algorithm will also be evaluated. We randomly select 19,000 data points as the training set and use the rest as the test set for prediction purpose. Then we implement the sampling method on the training dataset and obtain the coefficient estimator $\hat{\beta}$. We now measure the estimation and prediction performance by \emph{relative mean squared error}: $||\hat{\beta}-\beta_0||^2/||\beta_0||^2$ and \emph{prediction relative squared error}: $||X\hat{\beta}-Y||^2/||X\beta_0-Y||^2$ which is calculated over test dataset, respectively. We repeat the process $500$ times for different subsample sizes and the median of the two criteria are recorded for each subsample size. The results are presented in Figure \ref{rl}. We also report the median running times of the four sampling methods over different subsample sizes in Figure \ref{rl_t}.

\begin{figure}
	\begin{subfigure}{0.55\textwidth}
		\centering
		\includegraphics[scale=0.45]{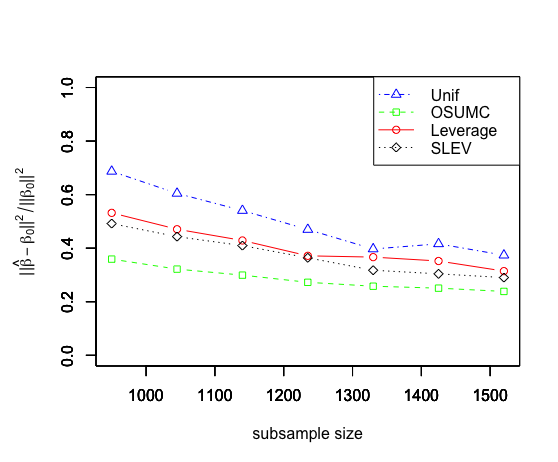}
		\caption{Estimation relative SE}
		\label{rl_1}
	\end{subfigure}%
	\begin{subfigure}{0.55\textwidth}
		\centering
		\includegraphics[scale=0.45]{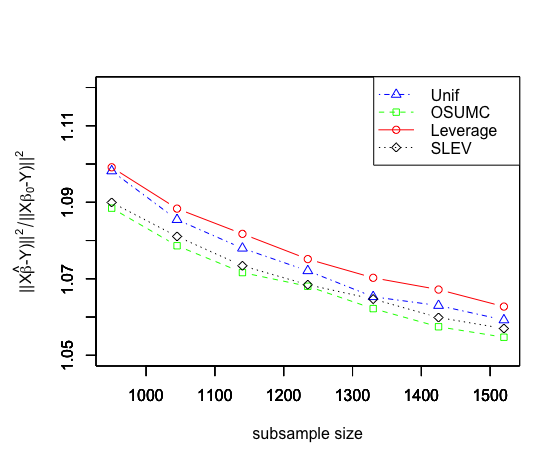}
		\caption{Prediction relative SE}
		\label{rl_2}
	\end{subfigure}
	\caption{Estimation and prediction performance comparison of four different sampling methods over different subsample size}\label{rl}
\end{figure}

Figure \ref{rl} shows that our method consistently achieves the best performance on both estimation and prediction. In addition, our sampling procedure outperforms both leverage-based methods in computational time. Though SLEV method achieves similar prediction accuracy as our method, it takes more than twice the computational time. Therefore, the proposed sampling procedure maintains a good balance between statistical efficiency and computational cost.

\section{Conclusion}\label{con}
In this paper, we propose a novel sampling procedure, OSUMC, to address measurement constraints when estimating GLMs. We show unconditional asymptotic normality of a general class of response-free sampling estimators by using newly developed martingale techniques. Our unconditional asymptotic results obtained without conditioning on the data are different from existing conditional results which condition on both covariates and responses. 
 Building on the asymptotic theory, we derive the A-optimal sampling distribution, which depends on population level quantities. For practical applications, we propose OSUMC algorithm to approximate the theoretical optimal sampling scheme. Additionally, we conduct extensive numerical studies which show that the performance of OSUMC is comparable to that of sampling algorithms which use complete responses to calculate sampling probabilities. This indicates that OSUMC successfully prevents the loss of statistical efficiency due to measurement constraints. 
 


A number of interesting extensions and open areas remain. For example, it would be interesting to consider the high-dimensional scenario where the dimension $p$ could be much larger than the subsample size $r$. Under such setting, techniques like regularization and debiasing in the high-dimensional statistics would be need which is out of the scope of this paper. We will leave it for future investigations. 

\section*{Supplementary Materials}
\textbf{Appendices} A supplementary PDF file with proofs of all the theoretical results in Section 4 and additional plots and simulation results for Section 5. \\
\textbf{R code}  Code for the simulations and real data analysis in Section 5.\\
\textbf{Dataset}  The Superconductivity dataset used for real data analysis.

\section*{Acknowledgements}
Ning was supported in part by NSF grant DMS-1854637.	
\bibliographystyle{apa}
\bibliography{Bibs2}

\setcounter{page}{1}
\setcounter{section}{0}
\setcounter{figure}{0}
\setcounter{table}{0}

\clearpage
\def\theequation{A\arabic{section}.\arabic{equation}}
\renewcommand{\thesection}{A\arabic{section}}   
\renewcommand{\thetable}{A\arabic{table}}   
\renewcommand{\thefigure}{A\arabic{figure}}
 \begin{center}
 {\Large\bf Supplementary Materials for ``Optimal Sampling for
Generalized Linear Models under Measurement Constraints"}
\vspace{.25cm}

{\large Tao Zhang, Yang Ning and David Ruppert}
\vspace{.4cm}

Department of Statistics and Data Science, Cornell University
\end{center}
\vspace{.55cm}

\section{Extra Notation}
 For a matrix $A\in \mathbb{R}^{p\times p}$, $\lambda_{\rm max}(A)$ denotes the largest eigenvalue of matrix $A$ and $||A||$ is the Frobenius norm of $A$. For two positive definite matrices, $B$ and $C$, we write $B>C$ if and only if $B-C$ is positive definite (p.d.). In particular, $B$ is p.d. if and only if $B > 0$.  

\section{Proof of Theorem \ref{Con}}
To show the consistency of $\hat{\beta}_n$, we will need the following lemma. 
\begin{lem}\label{Conlm1}
	For any $\beta\in \mathcal{B}$, if $\zeta(\beta):=E\left[\{b'(X^T\beta)-Y\}X\right]$ is finite and the following condition holds
	\begin{equation}\label{as1}
		\sum\limits_{i=1}^n E\left[ \frac{\{b'(X_i^T\beta)-Y_i\}^2}{\pi_i} x_{ij}^2\right]=o(n^2r),   
	\end{equation}
	then we have $ \Psi_n^*(\beta)-\zeta(\beta) \mathop{\longrightarrow}\limits^p 0$ for any $\beta\in \mathcal{B}$.
\end{lem}
\begin{proof}
Observe that 
\begin{align*}
E\Psi_n^*(\beta)&=E\left[\frac{1}{r} \sum_{i=1}^r\frac{b^{'}(X_i^{*T}\beta)-Y_i^*}{n\pi_i^*}\cdot X_i^*\right]\\
&=E\left[E\left\{\frac{1}{r} \sum_{i=1}^r\frac{b^{'}(X_i^{*T}\beta)-Y_i^*}{n\pi_i^*}\cdot X_i^*|(X_i,Y_i)_{i=1}^n\right\}\right]\\
&=\zeta.
\end{align*}
We will use Chebyshev's inequality to show convergence in probability. We denote the $j$-th element in the vector $(\Psi_n^*)$ by
\begin{align*}
(\Psi_n^*)_{j}&=\frac{1}{r} \sum_{i=1}^r\frac{b^{'}(X_i^{*T}\beta)-Y_i^*}{n\pi_i^*}\cdot x_{ij}^*,
\end{align*}
and the $j$-th coordinate of $\zeta$ as $\zeta_j$.

By using Chebyshev's inequality, it suffices to show $E[(\Psi_n^*)_{j}-\zeta_{j}]^2=o(1)$, where 
\begin{align*}
E[(\Psi_n^*)_{j}-\zeta_{j}]^2&=E\left[E\left\{\left[(\Psi_n^*)_{j}-\zeta_{j}\right]^2|(X_i,Y_i)_{i=1}^n\right\}\right].
\end{align*}
For this expectation
\begin{align*}
E\left[E\left\{\left[(\Psi_n^*)_{j}-\zeta_{j}\right]^2|(X_i,Y_i)_{i=1}^n\right\}\right]&=E\left[\frac{1}{r^2}\cdot \sum_{i=1}^rE_*\left\{\frac{b^{'}(X_i^{*T}\beta)-Y_i^*}{n\pi_i^*}\cdot x_{ij}^*-\zeta_{j}\right\}^2\right]\\
&=E\left[\frac{1}{r}\sum_{i=1}^n\pi_i\left\{\frac{b^{'}(X_i^{T}\beta)-Y_i}{n\pi_i}\cdot x_{ij}-\zeta_{j}\right\}^2\right]	\\
&=\frac{1}{rn^2}\sum_{i=1}^nE\left[\frac{\left\{b^{'}(X_i^{T}\beta)-Y_i\right\}^2}{\pi_i}\cdot x_{ij}^2\right]-\frac{1}{r}\zeta_{j}^2,
\end{align*}
where the first equality is based on the fact that after conditioning on the $n$ data points, the $r$ repeating sampling steps should be independent and distributionally identical in each step. And we use $E_*$ to denote expectation with respect to sampling randomness. Hence, we have$$E[(\Psi_n^*)_{j}-\zeta_{j}]^2=\frac{1}{rn^2}\sum_{i=1}^nE\left[\frac{(b^{'}\left( X_i^{T}\beta\right) -Y_i)^2}{\pi_i}\cdot x_{ij}^2\right]-\frac{1}{r}\zeta_{j}^2=o(1). $$
The second equality is due to the assumption (\ref{as1}).
\end{proof}
We now show the consistency of $\hat{\beta}_n$.

We will verify the conditions in Theorem 5.9 in \cite{Van2000} and apply the theorem to prove the consistency of $\hat{\beta}_n$.

First of all, for any $\beta_1$ and $\beta_2$ in $\mathcal{B}$,
\begin{align*}
&\left| [\Psi_n^*(\beta_1)-\zeta(\beta_1)]-[\Psi_n^*(\beta_2)-\zeta(\beta_2)]\right| \\
=&\left\| \left\{\frac{1}{r} \sum_{i=1}^r\frac{b^{''}(X_i^{*T}\tilde{\beta}_1)}{n\pi_i^*}\cdot X_i^*X_i^{*T}-E\left[b^{''}(X^T\tilde{\beta}_2)XX^T\right] \right\}\cdot(\beta_1-\beta_2) \right\| \\
\le& \left\|\frac{1}{r} \sum_{i=1}^r\frac{b^{''}(X_i^{*T}\tilde{\beta}_1)}{n\pi_i^*}\cdot X_i^*X_i^{*T}-E\left[b^{''}(X^T\tilde{\beta}_2)XX^T\right]\right\| \cdot \left\|\beta_1-\beta_2\right\|\\
:=&L_n\cdot  ||\beta_1-\beta_2||.
\end{align*}
The first step is due to mean value theorem with $\tilde{\beta}_1$ and $\tilde{\beta}_2$ lying on the segment between $\beta_1$ and $\beta_2$. 

We now show $L_n=O_p(1)$. By assumption (\ref{Con1}), it suffices to show $\frac{1}{r} \sum\limits_{i=1}^r\frac{X_i^*X_i^{*T}}{n\pi_i^*}= O_p(1)$. This is true because we have
\begin{align*}
&E\left[\frac{1}{r} \sum\limits_{i=1}^r\frac{X_i^*X_i^{*T}}{n\pi_i^*}\right]\\
=&E\left[E\left\{\frac{1}{r} \sum_{i=1}^r\frac{X_i^*X_i^{*T}}{n\pi_i^*}\bigg| (X_i,Y_i)_{i=1}^n\right\}\right]\\
=&EXX^T,
\end{align*} 
and it follows from Markov inequality.
Now we apply Lemma 2.9 in \cite{NEWEY19942111} to conclude $\Psi^*_n(\beta)-\zeta(\beta)$ is stochastic equicontinuous. Again, by Theorem 21.9 in \cite{Davidson1994}, Lemma \ref{Conlm1} and stochastic equicontinuity imply 
$$\sup\limits_{\beta\in \mathcal{B}} \left\|  \Psi^*_n(\beta)-\zeta(\beta)\right\| \mathop{\longrightarrow}\limits^p 0.$$
This uniform convergence condition together with condition (\ref{Con4}) in the theorem yield the desired conclusion by applying Theorem 5.9 in \cite{Van2000}. \qed
\section{Proof of Theorem \ref{AN}}
In this section, we will  establish the asymptotic normality of $\hat{\beta}_n$. Let us start will some auxiliary lemmas.

\subsection{Proof of Lemma \ref{AL}}	

We first prove following lemma.
\begin{lem}\label{ANlm1}
	Assume that $\Phi=E\left\{b^{''}(X^T\beta_0)XX^T\right\}$ is finite and non-singular. Further assume  for $1\le k,j \le p$, 
	\begin{equation}\label{as2}
		\sum_{i=1}^{n}E\left\{\frac{b^{''}(X_i^T\beta_0)^2}{\pi_i}(x_{ik}x_{ij})^2\right\}=o(n^2r).
	\end{equation}

	Then, we have $\dot{\Psi}_n^*(\beta_0)=\frac{1}{r} \sum_{i=1}^r\frac{b^{''}(X_i^{*T}\beta_0)}{n\pi_i^*}\cdot X_i^*X_i^{*T} \mathop{\longrightarrow}\limits^p \Phi$.
\end{lem} 
\begin{proof}
First we note
\begin{align*}
E\dot{\Psi}_n^*(\beta_0)&=E\left[\frac{1}{r} \sum_{i=1}^r\frac{b^{''}(X_i^{*T}\beta_0)}{n\pi_i^*}\cdot X_i^*X_i^{*T}\right]\\
&=E\left[E\left\{\frac{1}{r} \sum_{i=1}^r\frac{b^{''}(X_i^{*T}\beta_0)}{n\pi_i^*}\cdot X_i^*X_i^{*T}\bigg|(X_i,Y_i)_{i=1}^n\right\}\right] =\Phi.
\end{align*}
To show the convergence in probability, we use Chebyshev's inequality. Consider each element in the matrix
\begin{align*}
(\dot{\Psi}_n^*)_{kj}&=\frac{1}{r} \sum_{i=1}^r\frac{b^{''}(X_i^{*T}\beta_0)}{n\pi_i^*}\cdot x_{ik}^*x_{ij}^{*},\\
\Phi_{kj}& = E\left[b^{''}(X^T\beta_0)x_kx_j\right]. 
\end{align*}
By using Chebyshev's inequality, it suffices to show $E\left[(\dot{\Psi}_n^*)_{kj}-\Phi_{kj}\right]^2=o(1)$.
\begin{align*}
E\left[(\dot{\Psi}_n^*)_{kj}-\Phi_{kj}\right]^2&=E\left[E\left\{\left[(\dot{\Psi}_n^*)_{kj}-\Phi_{kj}\right]^2\bigg|(X_i,Y_i)_{i=1}^n\right\}\right].
\end{align*}
For this expectation
\begin{align*}
E\left[E\left\{\left[(\dot{\Psi}_n^*)_{kj}-\Phi_{kj}\right]^2\bigg|(X_i,Y_i)_{i=1}^n\right\}\right]&=E\left[\frac{1}{r^2}\cdot \sum_{i=1}^rE_*\left[\frac{b^{''}(X_i^{*T}\beta_0)}{n\pi_i^*}\cdot x_{ik}^*x_{ij}^{*}-\Phi_{kj}\right]^2\right]\\
&=E\left[\frac{1}{r}\sum_{i=1}^n\pi_i\left[\frac{b^{''}(X_i^{T}\beta_0)}{n\pi_i}\cdot x_{ik}x_{ij}-\Phi_{kj}\right]^2\right]	\\
&=\frac{1}{rn^2}\sum_{i=1}^nE\left[\frac{b^{''}(X_i^{T}\beta_0)^2}{\pi_i}\cdot x_{ik}^2x_{ij}^2\right]-\frac{1}{r}\Phi_{kj}^2,
\end{align*}
where the first equality is based on the fact that after conditioning on the $n$ data points, the $r$ repeating sampling steps should be independent and distributionally identical in each step, and we use $E_*$ to denote expectation with respect to sampling randomness. Hence, we have$$E\left[(\dot{\Psi}_n^*)_{kj}-\Phi_{kj}\right]^2=\frac{1}{rn^2}\sum_{i=1}^nE\left[\frac{b^{''}(X_i^{T}\beta_0)^2}{\pi_i}\cdot x_{ik}^2x_{ij}^2\right]-\frac{1}{r}\Phi_{kj}^2=o(1). $$
The second equality is due to the assumption (\ref{as2}).	
\end{proof}
Now we prove Lemma \ref{AL}.

By Taylor's Theorem:
$$0=\Psi_n^*(\hat{\beta}_n)=\Psi_n^*(\beta_0)+\dot{\Psi}^*_n(\beta_0)(\hat{\beta}_n-\beta_0)+\frac{1}{2}(\hat{\beta}_n-\beta_0)^T\ddot{\Psi}_n^*(\tilde{\beta}_n)(\hat{\beta}_n-\beta_0),$$
where $\tilde{\beta}_n$ is on the line segment between $\beta_0$ and $\hat{\beta}_n$. $\ddot{\Psi}_n^*$ is a k-vector of $(k\times k)$ matrices.

We now show that $\left\|\ddot{\Psi}^*_n(\tilde{\beta}_n)\right\|=O_p(1)$.
By  assumption (\ref{AL3})
\begin{align*}
\left\|\ddot{\Psi}^*_n(\tilde{\beta}_n)\right\| &=\left\|\frac{1}{r}\sum_{i=1}^r\frac{1}{n\pi_i^*}\cdot \ddot{\psi}_{\tilde{\beta}_n}(X_i^*)\right\|\\
&\le \frac{1}{r}\sum_{i=1}^r\frac{1}{n\pi_i^*}\cdot \left\|\ddot{\psi}(X_i^*)\right\|=O_p(1).
\end{align*}
The last equality is because of the fact $$E\left[\frac{1}{r}\sum_{i=1}^r\frac{1}{n\pi_i^*}\cdot \left\|\ddot{\psi}(X_i^*)\right\|\right]=E\left\|\ddot{\psi}(X)\right\|={\rm const}$$
and application of Markov inequality.

Therefore,	
$$	0=\Psi_n^*(\beta_0)+(\Phi+o_p(1))\left(\hat{\beta}_n-\beta_0\right)+O_p\left(\left\|\hat{\beta}_n-\beta_0\right\|^2\right).$$
This implies the conclusion
$$\Psi_n^*(\beta_0)=-\Phi(\hat{\beta}_n-\beta_0)+o_p\left(\left\|\hat{\beta}_n-\beta_0\right\|\right).$$ \qed

\subsection{Multivariate martingale CLT}
Now, we prove a multivariate extension of the martingale central limit theorem stated in \cite{Ohlsson1989} Theorem A.1, which will be appropriate for our with replacement sampling setting. 
\begin{lem}[Multivariate version of martingale CLT] \label{MCLT}
	For $k=1,2,3,\dots $, let $\{\xi_{ki};i=1,2,\dots ,N_k\}$ be a martingale difference sequence in $\mathbb{R}^p$ relative to the filtration $\{\mathcal{F}_{ki};i=0,1,\dots ,N_k\}$ and let $Y_k\in \mathbb{R}^p$ be an $\mathcal{F}_{k0}$-measurable random vector. Set $S_k=\sum\limits_{i=1}^{N_k}  \xi_{ki}$. Assume the following conditions.
	\begin{enumerate}[(i)]
		\item $\lim\limits_{k\to \infty}\sum\limits_{i=1}^{N_K}E\left[ \left\|\xi_{ki}\right\|^4 \right]=0$
		\item $\lim\limits_{k\to \infty}E\left[\left\|\sum\limits_{i=1}^{N_k}E\left[\xi_{ki}\xi_{ki}^T|\mathcal{F}_{k,i-1}\right]-B_k\right\|^2\right]=0$ for some sequence of positive definite matrices $\{B_k\}_{k=1}^{\infty}$ with $\sup\limits_k \lambda_{max}(B_k)<\infty$ i.e. the largest eigenvalue is uniformly bounded.
		\item For some probability distribution $L_0$, $*$ denotes convolution and $L(\cdot)$ denotes the law of random variables:
		$$L(Y_k)*N(0,B_k)\mathop{\longrightarrow}\limits^d L_0.$$ 
	\end{enumerate}
	Then we have
	$$L(Y_k+S_k)\mathop{\longrightarrow}\limits^d L_0.$$  
\end{lem}
\begin{proof}
We use Cramer-Wold device to deduce it from the univariate case. For any $a\in \mathbb{R}^p$, by Cramer-Wold device, it suffices to show 
$$L(a^TY_k+a^TS_k)\mathop{\longrightarrow}\limits^d a^TL_0.$$
We check the conditions of Theorem A.1 in \cite{Ohlsson1989}.
\begin{enumerate}
	\item $\sum\limits_{i=1}^{N_k} E\left[\left(a^T\xi_{ki}\right)^4\right]\le \sum\limits_{i=1}^{N_k} ||a||^4 \cdot E\left[||\xi_{ki}||^4\right]=||a||^4\sum\limits_{i=1}^{N_k}  E[||\xi_{ki}||^4] \longrightarrow 0 $ The inequality is due to Cauchy-Schwarz inequality.
	\item \begin{align*}
	&E\left[\sum\limits_{i=1}^{N_k}E\left[a^T\xi_{ki}\xi_{ki}^Ta|\mathcal{F}_{k,i-1}\right]-a^TB_ka\right]^2\\
	=&E\left[a^T\left\{\sum\limits_{i=1}^{N_k}E[\xi_{ki}\xi_{ki}|\mathcal{F}_{k,i-1}]-B_k\right\}a\right]^2\\
	\lesssim&E\left[\left\|\sum\limits_{i=1}^{N_k}E[\xi_{ki}\xi_{ki}|\mathcal{F}_{k,i-1}]-B_k\right\|\cdot ||a||^2\right]^2 \longrightarrow 0\   \text{as}\  k \to \infty.
	\end{align*}
	\item \begin{align*}
	\phi_{a^TY_k}\cdot\phi_{N(0,a^TB_ka)}=&E\left[e^{ita^TY_k}\right]\cdot e^{-\frac{1}{2}(a^TB_ka)t^2}\\
	=&E\left[e^{i\xi^TY_k}\right]\cdot e^{-\frac{1}{2}\xi^TB_k\xi}\qquad \text{where}\ \xi=at\\
	\longrightarrow &\phi_{L_0}(at)\equiv \phi_{a^TL_0}(t).
	\end{align*}
	Here we use $\phi_{*}(t)$to denote the characteristic function. Hence
	$$L(a^TY_k)*N(0,a^TB_ka)\mathop{\longrightarrow}\limits^d a^TL_0.$$ 
\end{enumerate}
From above verification, we use Theorem A.1 in \cite{Ohlsson1989} to obtain
$$L(a^TY_k+a^TS_k)\mathop{\longrightarrow}\limits^d a^TL_0.$$
And by Cramer-Wold device, this finishes the proof.	
\end{proof}

\subsection{More Auxiliary Results}
\begin{lem} \label{lem mg}
	$\{M_i\}_{i=1}^r$ is a martingale difference sequence relative to the filtration $\{\mathcal{F}_{n,i}\}_{i=1}^r$.	
\end{lem}
\begin{proof}
	The $\mathcal{F}_{n,i}$-measurability follows from the definition of $M_i$ and the definition of the filtration $\{\mathcal{F}_{n,i}\}_{i=1}^r$. And we also have
	\begin{align*}
	E[M_i|{\mathcal{F}_{n,i-1}}]&= E_{*_i}	\left[\frac{b'(X_i^{*T}\beta)-Y_i^*}{rn\pi_i^*}\cdot X_i^*\right]-\frac{1}{rn}\sum\limits_{j=1}^n\left(b'(X_j^T\beta_0)-Y_j\right)\cdot X_j =0.
	\end{align*}
	Combine these two results, we finish the proof.	
\end{proof}
With Lemma 4, we could easily get the following result.
\begin{cor}
	$V(T)=V(M)+V(Q).$
\end{cor}
\begin{lem}
	$\sup\limits_n \lambda_{max}(B_n)\le 1.$
\end{lem}
\begin{proof}
	Since $B_n$ is symmetric, it suffices to show for any $n$, $I-B_n$ is positive definite.
	\begin{align*}
	I-B_n&=V(T)^{-\frac{1}{2}}\left(V(T)-V(M)\right)V(T)^{-\frac{1}{2}}\\
	&=V(T)^{-\frac{1}{2}}V(Q)V(T)^{-\frac{1}{2}}.
	\end{align*}
	Therefore, $I-B_n$ is congruent to matrix $V(Q)$ which is positive definite. Hence $I-B_n$ is also positive definite and this finishes the proof.
\end{proof}
\begin{lem}[Asymptotic normality of $\Psi_n^*(\beta_0)$]\label{ANPSI}
	Assume the following conditions
	\begin{enumerate}[(i)]
		\item $\lim\limits_{n\to \infty} \sum\limits_{i=1}^rE\left[||\xi_{ni}||^4\right]=0.$ \label{as3}
		\item $\lim\limits_{n\to \infty}E\left[\left\|\sum\limits_{i=1}^{r}E[\xi_{ni}\xi_{ni}^T|\mathcal{F}_{n,i-1}]-B_n\right\|^2\right]=0.$ \label{as4}
	\end{enumerate}
	Then we will have
	$$V(T)^{-\frac{1}{2}}\cdot T \mathop{\longrightarrow}\limits^d N(0,I).$$
\end{lem}
\begin{proof}
	We verify the conditions in Lemma \ref{MCLT} with
	\begin{align*}
	&\xi_{ki}=\xi_{ni}, \qquad Y_k=V(T)^{-\frac{1}{2}}\cdot Q,\\
	&B_k=B_n,  \qquad  L_0 \sim N(0,I).
	\end{align*}
	By Lemma \ref{lem mg}, conditions (\ref{as3}) and (\ref{as4}), we can easily see the first two conditions of Lemma \ref{MCLT} are satisfied. It suffices to show the third condition in Lemma \ref{MCLT} holds. We first note the following conclusion
	$$V(Q)^{-\frac{1}{2}}Q \mathop{\longrightarrow}\limits^d N(0,I).$$
	This because $Q$ is a sum of i.i.d mean zero random variables, $(b'(X_j^T\beta_0)-Y_j)\cdot X_j$, which have finite variance and a simple application of central limit theorem will give the above conclusion.
	
	Now, we verify the third condition. For any $t\in \mathbb{R}^p$
	
	\begin{align*}
	&E[e^{it^TV(T)^{-\frac{1}{2}}Q}] \cdot e^{-\frac{1}{2}t^TV(T)^{-\frac{1}{2}}V(M)V(T)^{-\frac{1}{2}}t}\\
	=&\left(Ee^{it^TV(T)^{-\frac{1}{2}}V(Q)V(T)^{-\frac{1}{2}}t}+o(1)\right)\cdot e^{-\frac{1}{2}t^TV(T)^{-\frac{1}{2}}V(M)V(T)^{-\frac{1}{2}}t}\\
	=&Ee^{it^TV(T)^{-\frac{1}{2}}V(Q)V(T)^{-\frac{1}{2}}t} \cdot e^{-\frac{1}{2}t^TV(T)^{-\frac{1}{2}}V(M)V(T)^{-\frac{1}{2}}t}+o(1)\\
	=& e^{-\frac{1}{2}t^Tt}+o(1).
	\end{align*}	
	The first equality is due to Lemma \ref{cd3} in the following. 	
	Therefore, we have verified the third condition in Lemma \ref{MCLT}. And by that lemma we have
	$$V(T)^{-\frac{1}{2}}\cdot Q+V(T)^{-\frac{1}{2}}\cdot M=V(T)^{-\frac{1}{2}}T \mathop{\longrightarrow}\limits^d N(0,I). $$
\end{proof}
Now we state the following lemma that has been used in the proof of previous lemma.
\begin{lem}\label{cd3}
	Under conditions in Lemma \ref{ANPSI} For any $t\in \mathbb{R}^p$,
	$$\left|E\left[e^{it^TV(T)^{-\frac{1}{2}}Q}\right]-E\left[e^{it^TV(T)^{-\frac{1}{2}}V(Q)^{\frac{1}{2}}A_0}\right]\right|\longrightarrow 0$$ as $n \to \infty$ and $A_0\sim N(0,I)$.
\end{lem}
\begin{proof}
	Since $V(Q)^{-\frac{1}{2}}Q \mathop{\longrightarrow}\limits^d N(0,I)$, for any $\xi\in \mathbb{R}^p$,
	$$\left|E\left[e^{i\xi ^TV(Q)^{-\frac{1}{2}}Q }\right]-E\left[e^{i\xi^TA_0}\right]\right| \longrightarrow 0$$
	as $n \to \infty$.	And the convergence is uniform in any finite set of $\xi$. (see Chapter 6 of \cite{Chung2001}). By setting $\xi=V(Q)^{\frac{1}{2}}V(T)^{-\frac{1}{2}}t^T$, to prove the lemma, it suffices to show 
	$$\sup\limits_n||\xi||< \infty.$$
	for any fixed $t$.
	Also we note that $$||\xi||\le \lambda_{max}\left(V(Q)^{\frac{1}{2}}V(T)^{-\frac{1}{2}}\right)\cdot ||t||.$$
	Hence, it is enough to show $\lambda_{max}\left(V(Q)^{\frac{1}{2}}V(T)^{-\frac{1}{2}}\right)\le 1$.
	For notation simplicity, we denote $A=V(Q)$ and $B=V(T)$ in the following proof of the lemma.Note the following equation holds
	$$A^{\frac{1}{2}}B^{-\frac{1}{2}}=B^{\frac{1}{4}}\left(B^{-\frac{1}{4}}A^{\frac{1}{2}}B^{-\frac{1}{4}}\right)B^{-\frac{1}{4}}$$
	That is $A^{\frac{1}{2}}B^{-\frac{1}{2}}$ is similar to $B^{-\frac{1}{4}}A^{\frac{1}{2}}B^{-\frac{1}{4}}$. Therefore, we only need to show $\lambda_{max}\left(B^{-\frac{1}{4}}A^{\frac{1}{2}}B^{-\frac{1}{4}}\right)\le 1$. This is implied by the fact $$I-B^{-\frac{1}{4}}A^{\frac{1}{2}}B^{-\frac{1}{4}}>0.$$
	In fact
	$$I-B^{-\frac{1}{4}}A^{\frac{1}{2}}B^{-\frac{1}{4}}=B^{-\frac{1}{4}}\left(B^{\frac{1}{2}}-A^{\frac{1}{2}}\right)B^{-\frac{1}{4}}.$$
	That is, $I-B^{-\frac{1}{4}}A^{\frac{1}{2}}B^{-\frac{1}{4}}$ is congruent to $B^{\frac{1}{2}}-A^{\frac{1}{2}}$. Therefore, it suffices to show $B^{\frac{1}{2}}-A^{\frac{1}{2}}$ is positive definite.
	
	Note $B>A$. This is because $B-A=V(M)>0$. Now we use Theorem 1.1 in \cite{Zhan2004}, then we will have $B^{\frac{1}{2}}-A^{\frac{1}{2}}>0$ which finishes proof.
\end{proof}
Now we are able to prove our theorem \ref{AN}  which shows the asymptotic normality of the sampling estimator $\hat{\beta}_n$.
\subsection{Proof of Theorem \ref{AN}}
\begin{proof}
	By Lemma \ref{AL}$$\Phi(\hat{\beta}_n-\beta_0)+o_p\left(\left\|\hat{\beta}_n-\beta_0\right\|\right)=-\Psi_n^*(\beta_0).$$
	Now we normalize both sides with $V(T)^{-\frac{1}{2}}$
	$$V(T)^{-\frac{1}{2}}\Phi(\hat{\beta}_n-\beta_0)+o_p\left(\left\|V(T)^{-\frac{1}{2}}\hat{\beta}_n-\beta_0\right\|\right)=-V(T)^{-\frac{1}{2}}\Psi_n^*(\beta_0).$$
	By Lemma \ref{ANPSI}
	$$V(T)^{-\frac{1}{2}}\Phi(\hat{\beta}_n-\beta_0)\mathop{\longrightarrow}\limits^d N(0,I).$$
\end{proof}
\section{Proof of Theorem \ref{thm2}}

First of all, we condition on $X_1^n$(or consider $X_1^n$ is fixed).
We now find out $V(T|X_1^n)=V(\Psi_n^*(\beta_0)|X_1^n)$. We have
$$V(T|X_1^n)=E_Y\left[V(T|X_1^n,Y_1^n)\right]+V_Y\left[E(T|X_1^n,Y_1^n)\right].$$
Here $E_Y$ means we take expectation w.r.t randomness of $Y$ after we conditioning on $X$. 
After some simple calculation, we could get
\begin{align*}
V_Y[E(T|X_1^n,Y_1^n)]&= \frac{1}{n^2}\sum\limits_{j=1}b^{''} (X_j^T\beta_0)\cdot X_jX_j^T,\\
E_Y[V(T|X_1^n,Y_1^n)]&=\frac{1}{n^2r}\sum\limits_{j=1}^nb^{''} (X_j^T\beta_0)\cdot X_jX_j^T\cdot \left(\frac{1}{\pi_j}-1\right).
\end{align*} 
Hence, we have
$$V(T|X_1^n)=\frac{1}{n^2}\sum\limits_{j=1}^nb^{''} (X_j^T\beta_0)\cdot X_jX_j^T\cdot \left(\frac{1}{r\pi_j}-\frac{1}{r}+1\right).$$
We now minimize $tr(\Phi^{-1}V(T|X_1^n)\Phi^{-1})$
\begin{align*}
tr(\Phi^{-1}V(T|X_1^n)\Phi^{-1})&=\frac{1}{n^2}\sum\limits_{j=1}^n tr\left(b^{''} (X_j^T\beta_0)\cdot \Phi^{-1}X_jX_j^T\Phi^{-1}\cdot (\frac{1}{r\pi_j}-\frac{1}{r}+1)\right)\\
&=\frac{1}{rn^2}\sum\limits_{j=1}^n tr\left(\frac{b^{''} (X_j^T\beta_0)}{\pi_j}\cdot \Phi^{-1}X_jX_j^T\Phi^{-1}\right)+ C\\
&=\frac{1}{rn^2}\sum\limits_{j=1}^n \frac{b^{''} (X_j^T\beta_0)}{\pi_j}\cdot \left\|\Phi^{-1}X_j\right\|^2+ C\\
&=\frac{1}{rn^2}\sum\limits_{j=1}^n\pi_j \sum\limits_{j=1}^n \frac{b^{''} (X_j^T\beta_0)}{\pi_j}\cdot \left\|\Phi^{-1}X_j\right\|^2+ C\\
&\ge\frac{1}{rn^2}\left(\sum\limits_{j=1}^n \sqrt{b^{''} (X_j^T\beta_0)}\cdot \left\|\Phi^{-1}X_j\right\|\right)^2+ C,
\end{align*}
where in the last step we use Cauchy-Schwarz inequality and the equality holds iff $\pi_j \propto \sqrt{b^{''} (X_j^T\beta_0)} \left\|\Phi^{-1}X_j\right\|$.

Now we consider $V(T)$ under random design.
$$V(T)=E\left[V(T|X_1^n)\right]+V\left[E(T|X_1^n)\right].$$
However, we could verify that in GLM
$$E(T|X_1^n)\equiv 0.$$
Therefore, we have $V(T)=E[V(T|X_1^n)]$. 
From this we have
\begin{align*}
\{\pi_j^{opt}\}_{j=1}^n&=\argmin\limits_{\pi}tr(\Phi^{-1} V(T)\Phi^{-1})\\
&=\argmin\limits_{\pi}tr\left(E\left[\Phi^{-1}V(T|X_1^n)\Phi^{-1}\right]\right)\\
&=\argmin\limits_{\pi}E\left[tr\left(\Phi^{-1}V(T|X_1^n)\Phi^{-1}\right)\right]\\
&=\argmin\limits_{\pi}tr\left(\Phi^{-1}V(T|X_1^n)\Phi^{-1}\right).
\end{align*}

\section{Additional Plots of Section \ref{sr}}
\subsection{Computational Time Plots for Logistic Regression}\label{logtm}
\begin{figure}[H]
	\centering
	\includegraphics[scale=0.52]{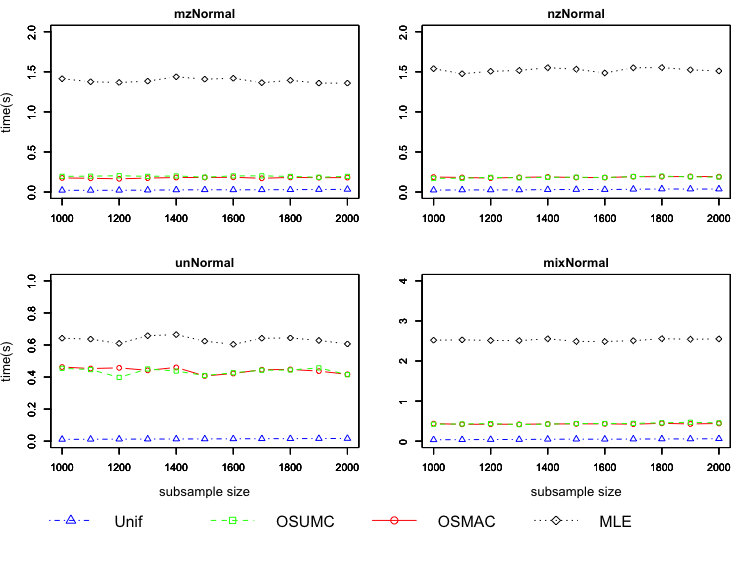}
	\caption{Computational time for different subsample size $r$ under different design generation settings for logistic regression with $r_0=500$}
	\label{logtime}
\end{figure}
 Figure \ref{logtime} reveals that the computation time is not very sensitive to the subsample size for all the four methods. All of the three sampling methods outperform the full sample \mbox{MLE}. It is not surprising to see the uniform sampling always takes the least computation time, since it does not involve the computation of the sampling probability to compensate for the loss of efficiency.  In most cases, OSUMC and OSMAC require significantly less computational time compared with full-sample \mbox{MLE}. 

\subsection{Q-Q Plots for Logistic Regression}\label{logqqs}
To see whether the asymptotic normality in our theory implies approximate finite-sample normality under the previous four different design generations in logistic regression model, we plot the chi-square Q-Q plot of the resultant estimator $\hat{\beta}_n$ for each considered setting.  Here, we replace the approximated optimal sampling weight in Algorithm \ref{ts} with the oracle optimal weights to calculate the estimator $\hat{\beta}_n$, i.e., the true $\beta_0$ is used in the calculation of optimal sampling weights. Experiments are repeated $1000$ times under each setting and corresponding Q-Q plots are presented in the Figure \ref{logqq}. 
Nearly all the points lie on the straight line in each plot, which is consistent with $\hat{\beta}_n$ being approximately normally distributed in the four considered design generation settings.
\begin{figure}[H]
	\begin{subfigure}{0.55\textwidth}
		\centering
		\includegraphics[scale=0.5]{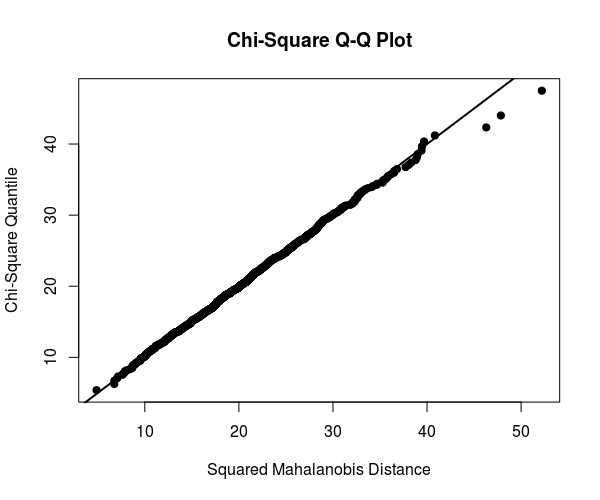}
		\caption{mzNormal}
		
	\end{subfigure}%
	\begin{subfigure}{0.55\textwidth}
		\centering
		\includegraphics[scale=0.5]{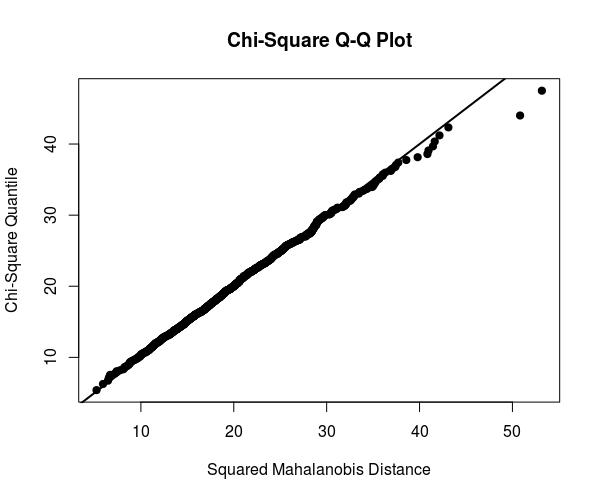}
		\caption{nzNormal}
		
	\end{subfigure}
	\begin{subfigure}{0.55\textwidth}
		\centering
		\includegraphics[scale=0.5]{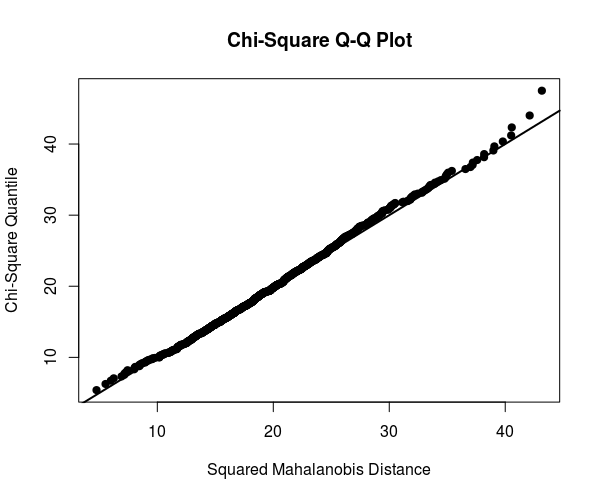}
		\caption{unNormal}
		
	\end{subfigure}%
	\begin{subfigure}{0.55\textwidth}
		\centering
		\includegraphics[scale=0.5]{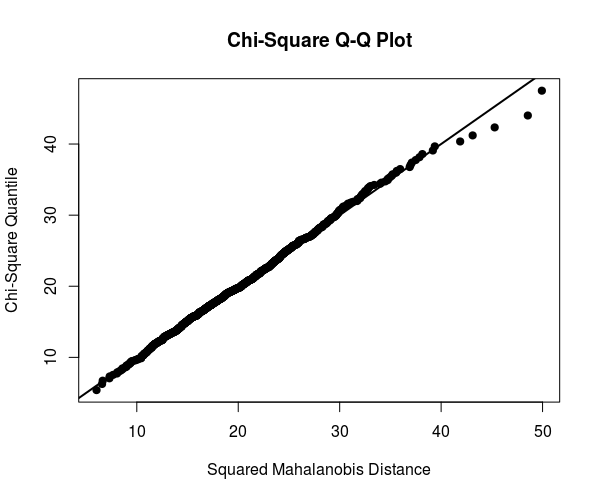}
		\caption{mixNormal}
		
	\end{subfigure}
	\caption{Chi-square Q-Q plots of $\hat{\beta}_n$ under different design generation settings for logistic regression with $r=5000$.}
	\label{logqq}
\end{figure}
\subsection{Computational Time Plots for Linear Regression}\label{lrct}

\begin{figure}[H]
	\centering
	\includegraphics[scale=0.42]{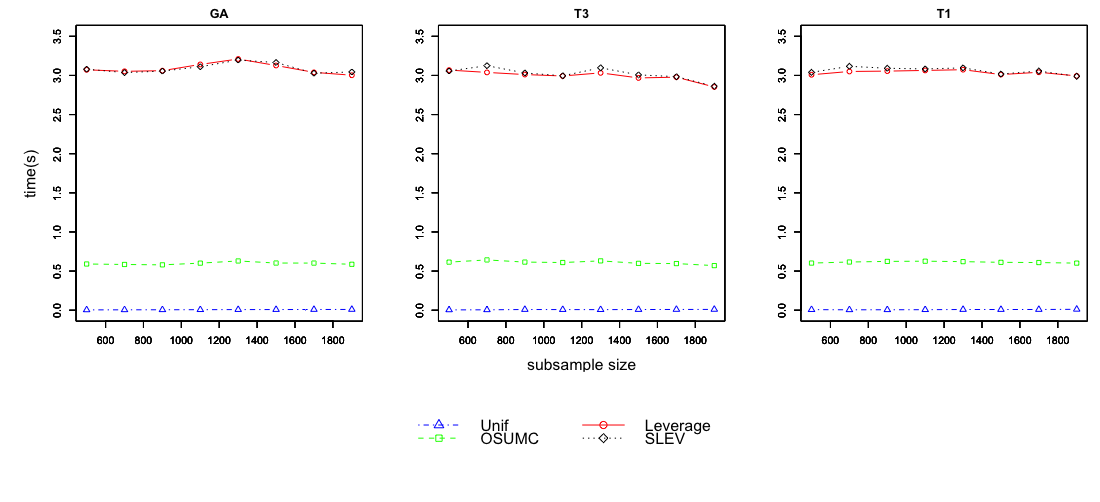}
	\caption{Computational time plots for different subsample size $r$ under different design generation settings for linear regression}
	\label{lrtime}
\end{figure}
From Figure \ref{lrtime}, Again, the results show the insensitivity of the computational time to increasing subsample sizes. Our method requires the second smallest computing time, being inferior only to the uniform sampling. Both leverage related methods take more than double the computational time of our method due to the intensive computation of leveraging score of each data point.

\subsection{Q-Q Plots for Linear Regression}\label{lrqqs}	
To explore further how sensitive the approximate finite-sample normality is to the moment condition of the design distribution in linear regression setting, we show chi-square Q-Q plots  for several design generation distributions with different orders of moment. To be more specific, GA, $T_9$, $T_3$ and $T_1$ distributions are considered. Experiments are repeated $1000$ times under each setting and results are presented in the following.

\begin{figure}[H]
	\begin{subfigure}{0.55\textwidth}
		\centering
		\includegraphics[scale=0.5]{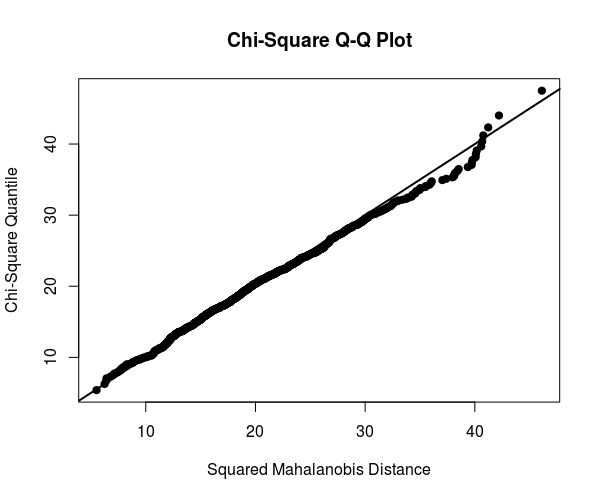}
		\caption{GA}
		
	\end{subfigure}%
	\begin{subfigure}{0.55\textwidth}
		\centering
		\includegraphics[scale=0.5]{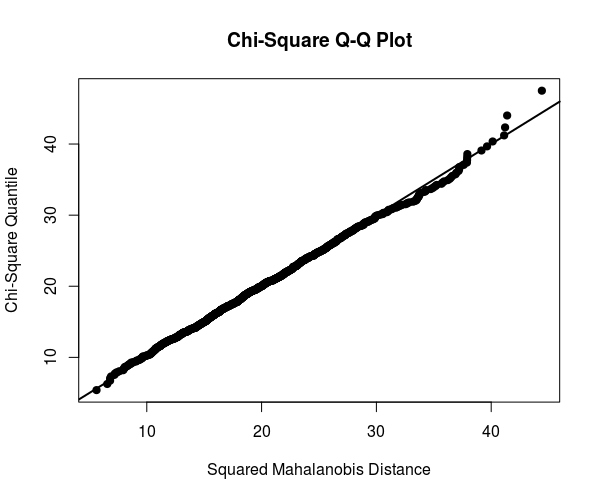}
		\caption{$T_9$}
		
	\end{subfigure}
	\begin{subfigure}{0.55\textwidth}
		\centering
		\includegraphics[scale=0.5]{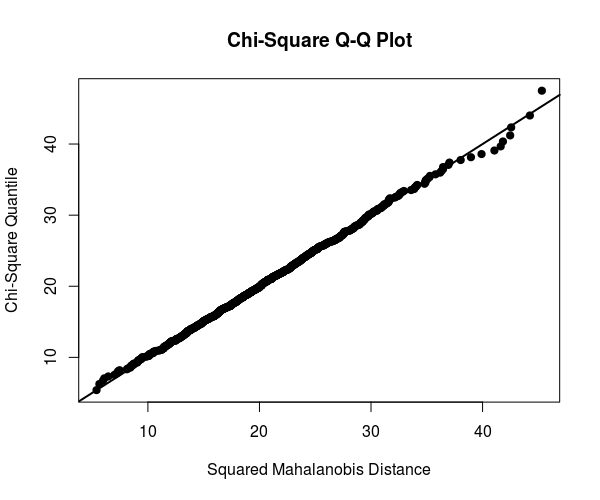}
		\caption{$T_3$}
		
	\end{subfigure}%
	\begin{subfigure}{0.55\textwidth}
		\centering
		\includegraphics[scale=0.5]{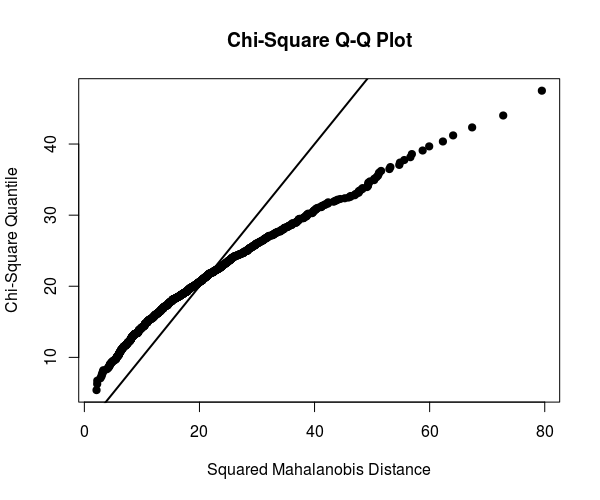}
		\caption{$T_1$}
		
	\end{subfigure}%
	\caption{Chi-square Q-Q plots of $\hat{\beta}_n$ under different design generation settings for linear regression with $r=5000$.}\label{qqlm}
\end{figure}
As shown in Figure \ref{qqlm}, the resultant sampling estimator $\hat{\beta}_n$ is approximately normal in GA, $T_9$ and $T_3$ settings where we should note that the multivariate t-distribution with 3 degrees of freedom doesn't even have a third moment. This indicates that the normality of $\hat{\beta}_n$  in linear models holds under very weak moment conditions for the design-generation distribution. One surprising fact is that OSUMC outperforms the other sampling methods in the $T_1$ setting, even though normality fails to hold.  

\subsection{Computational Time Plot for Superconductivity Data Set}\label{rdct}		
	\begin{figure}[H]
		\centering
		\includegraphics[scale=0.6]{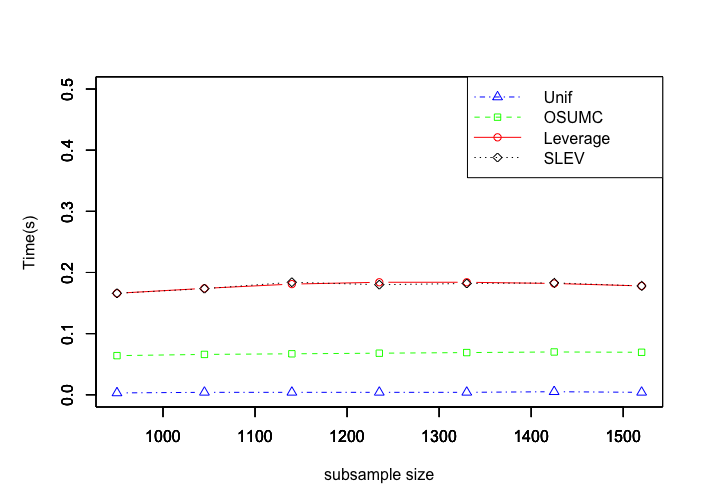}
		\caption{Computational time plots for different subsample size}\label{rl_t}
	\end{figure}

\section{Simulation for Poisson Regression}\label{Poisson}
We generate datasets of size $n=100,000$ from the following Poisson regression model,
$$Y\sim \rm{Poisson}\left(\exp\{X^T\beta_0\}\right),$$
where $\beta_0$ is a $100$ dimensional vector with all entries $0.5$. We consider two different scenarios to generate $X$. 
\begin{itemize}
	\item \textbf{Case 1.} each covariate of $X$ follows independent uniform distribution over $[-0.5,0.5]$.
	\item  \textbf{Case 2.} First half of covariates of $X$ follow independent uniform distribution over $[-0.5,0.5]$ while the other half follow independent uniform distribution over $[-1,1]$.
\end{itemize}
In each case, we compare our optimal sampling procedure (OSUMC) with uniform sampling (Unif), and the benchmark full data MLE under different subsample sizes. In our procedure, $r_0$, the subsample size in the first step uniform sampling, equals $500$. For uniform sampling, we directly subsample $r$ points and calculate the subsample MLE. Again, we repeat the simulation $500$ times and report the empirical MSE and computational time in Figures \ref{poimse} and \ref{poitime}, respectively.  
\begin{figure}[h]
	
	\centering
	\includegraphics[scale=0.7]{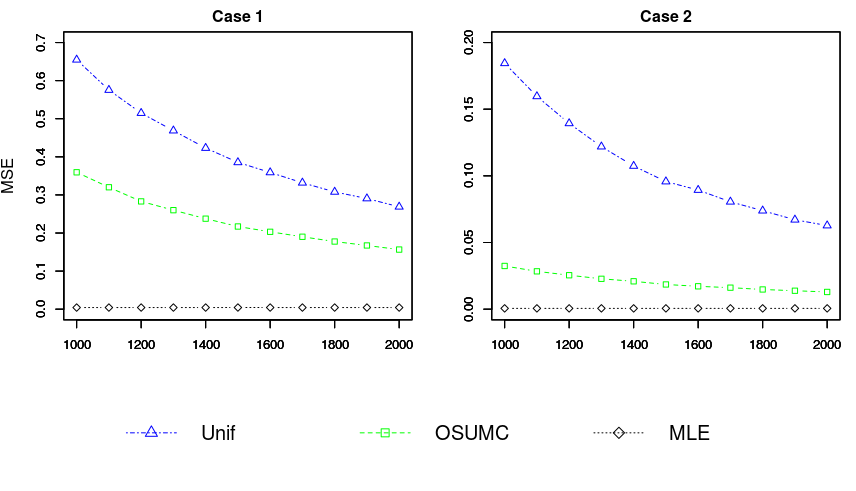}
	\caption{MSE of the proposed optimal sampling procedure (OSUMC), the uniform sampling (Unif), and the full sample MLE (MLE) for different subsample size $r$ under two scenarios in Poisson regression.}
	\label{poimse}
\end{figure}
\begin{figure}[h]
	
	\centering
	\includegraphics[scale=0.7]{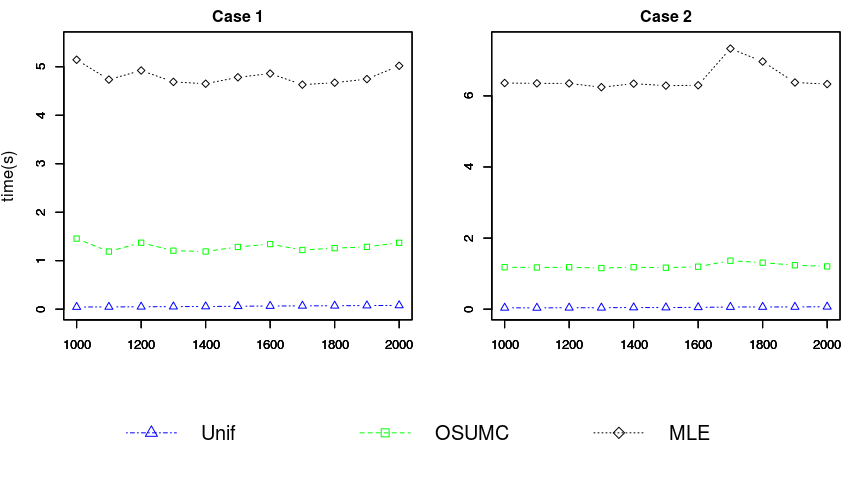}
	\caption{Computational time plot for different subsample size under different design generation settings for Poisson regression.}
	\label{poitime}
\end{figure}

From Figure \ref{poimse}, OSUMC method uniformly dominates the uniform sampling method in both scenarios in terms of mean squared errors, which supports the A-optimality of OSUMC. For average computational time, our simulation reveals that the computation time is not very sensitive to the subsample size. In both cases, OSUMC requires significantly less computational time compared with full-sample MLE. 

To see whether the asymptotic normality in our theory holds under the both design generation settings, we plot the chi-square Q-Q plot of the resultant estimator $\hat{\beta}_n$ for each considered setting and the results are presented in Figure \ref{poiqq}. Q-Q plots reveal that $\hat{\beta}_n$ is approximately normal with sufficiently large sample size $n$ and subsample size $r$ in the both considered design generation settings, which again provides empirical support for our theoretical results, especially the discussion of Poisson regression after Theorem \ref{Con}.  
\begin{figure}[H]
	\begin{subfigure}{0.55\textwidth}
		\centering
		\includegraphics[scale=0.55]{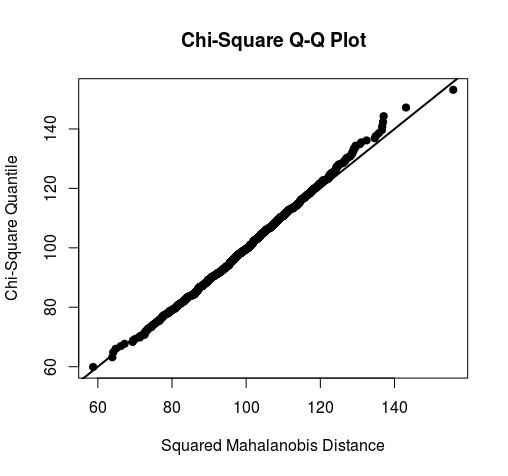}
		\caption{Case 1}
		
	\end{subfigure}%
	\begin{subfigure}{0.55\textwidth}
		\centering
		\includegraphics[scale=0.55]{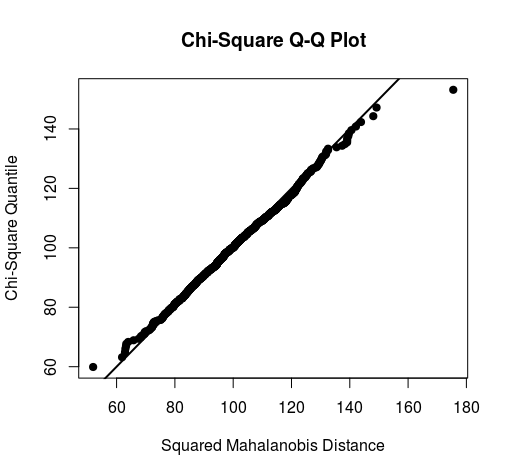}
		\caption{Case 2}
		
	\end{subfigure}
	\caption{Chi-square Q-Q plots of $\hat{\beta}_n$ under different design generation settings for Poisson regression with $r=5000$.}
	\label{poiqq}
\end{figure}

\bibliographystyle{apa}
\bibliography{Bibs2} 

\end{document}